\documentclass[runningheads]{llncs}
\usepackage[utf8]{inputenc}
\usepackage{amsfonts}
\usepackage{amsmath}

\usepackage{amsthm}
\usepackage{graphicx}
\usepackage{wrapfig}
\usepackage{subfig}
\usepackage{array}
\usepackage{tabularx}

\newcommand{\tuple}[1]{\langle #1 \rangle}
\newcommand{\nat}{\mathbb{N}}

\newtheorem{observation}{Observation}
\newtheorem{net}{Network}

\title{Smaller Selection Networks for Cardinality Constraints Encoding}

\titlerunning{Smaller Selection Networks for Cardinality Constraints Encoding}

\author{
Micha\l~Karpi\'nski, Marek Piotr\'ow}

\authorrunning{Micha\l~Karpi\'nski \and  Marek Piotr\'ow}

\date{}

\institute{Institute of Computer Science, University of Wroc\l aw\\
  Joliot-Curie 15, 50-383 Wroc\l aw, Poland\\
\email{\{karp,mpi\}@cs.uni.wroc.pl}}

\begin{document}

\pagestyle{headings}

\maketitle

\begin{abstract}
  \noindent Selection comparator networks have been studied for many years.
  Recently, they have been successfully applied to encode cardinality constraints
  for SAT-solvers. To decrease the size of generated formula there is a need for
  constructions of selection networks that can be efficiently generated and
  produce networks of small sizes for the practical range of their two parameters:
  $n$ -- the number of inputs (boolean variables) and $k$ -- the number of
  selected items (a cardinality bound). In this paper we give and analyze a new
  construction of smaller selection networks that are based on the pairwise
  selection networks introduced by Codish and Zanon-Ivry. We prove also that 
  standard encodings of cardinality constraints with selection networks preserve 
  arc-consistency.
\end{abstract}

\section{Introduction} Comparator networks are probably the simplest
data-oblivious model for sorting-related algorithms. The most popular construction
is due to Batcher \cite{batcher} and it's called {\em odd-even} sorting network.
For all practical values, this is the best known sorting network. However, in 1992
Parberry \cite{parberry} introduced the serious competitor to Batcher's
construction, called {\em pairwise} sorting network. In context of sorting,
pairwise network is not better than odd-even network, in fact it has been proven
that they have exactly the same size and depth. As Parberry said himself: {\em''It
is the first sorting network to be competitive with the odd-even sort for all
values of n``}. There is a more sophisticated relation between both types of
network and their close resemblance. For overview of sorting networks, see Knuth
\cite{knuth} or Parberry \cite{parberry2}.

In recent years new applications for sorting networks have been found, for example
in encoding of {\em pseudo boolean constraints} and {\em cardinality constraints}
for SAT-solvers. Cardinality constraints take the form $x_1+x_2+\cdots + x_n \sim
k$, where $x_1,x_2,\ldots,x_n$ are boolean variables, $k$ is a natural number, and
$\sim$ is a relation from the set $\{=,<,\leq,>,\geq\}$. Cardinality constraints
are used in many applications, the significant one worth mentioning arise in SAT-solvers.
Using cardinality constraints with cooperation of SAT-solvers we can
handle many practical problems that are proven to be hard. Works of As\'in {\em et
al.} \cite{card1,card4} describe how to use odd-even sorting network to encode
cardinality constraints into boolean formulas. In \cite{card3} authors do the same
with pseudo boolean constraints.

It has already been observed that using selection networks instead of sorting
networks is more efficient for the encoding of cardinality constraints. Codish and
Zazon-Ivry \cite{card2} introduced pairwise cardinality networks, which are
networks derived from pairwise sorting networks that express cardinality
constraints. Two years later, same authors \cite{pairwise} reformulated the 
definition of pairwise selection networks and proved that their sizes are never 
worse than the sizes of corresponding odd-even selection networks. To show the 
difference they plotted it for selected values of $n$ and $k$.
  
In this paper we give a new construction of smaller selection networks that are 
based on the pairwise selection ones and we prove that the construction is correct. 
We estimate also the size of our networks and compute the difference in sizes 
between our selection networks and the corresponding pairwise ones. The difference 
can be as big as $n\log n / 2$ for $k = n/2$. Finally, we analyze the standard 
3(6)-clause encoding of a comparator and prove that such CNF encoding of any 
selection network preserves arc-consistency with respect to a corresponding 
cardinality constraint.  

The rest of the paper is organized in the following way: in Section 2 we give
definitions and notations used in this paper. In Section 3 we recall the
definition of pairwise selection networks and define auxiliary bitonic selection
networks that we will use to estimate the sizes of our networks. In Section 4 we
present the construction of our selection networks and prove its correctness. In
Section 5 we analyze the sizes of the networks and, finally, in Section 6 we
examine the arc-consistency of selection networks.
  
\section{Preliminaries}

In this section we will introduce definitions and notations used in the rest of
the paper.

  \begin{definition}[input sequence]
  Input sequence of length $n$ is a sequence of natural numbers $\bar{x}~=~\tuple{x_1,\ldots,x_n}$,
  where $x_i \in \nat$ (for all $i=1..n$). We say that $\bar{x} \in \nat^n$ is sorted if $x_i \geq x_{i+1}$ (for each $i=1..n-1$).
  Given $\bar{x}~=~\tuple{x_1,\ldots,x_n}$, $\bar{y}~=~\tuple{y_1,\ldots,y_n}$ we define concatenation as
  $\bar{x} :: \bar{y} = \tuple{x_1,\ldots,x_n,y_1,\ldots,y_n}$. We will use the following functions from $\nat^n$ to $\nat^{n/2}$:

  \[ 
    left(\bar{x})=\tuple{x_1,\ldots,x_{n/2}}, \quad\quad right(\bar{x})=\tuple{x_{n/2+1},\ldots,x_{n}} 
  \]

  \noindent Let $n,m \in \nat$. We define a relation '$\succeq$' on $\nat^n \times \nat^m$.
  Let $\bar{x}~=~\tuple{x_1,\ldots,x_n}$ and $\bar{y}~=~\tuple{y_1,\ldots,y_m}$, then:

  \[
    \bar{x} \succeq \bar{y} \iff \forall_{i \in \{1,\ldots,n\}} \forall_{j \in \{1,\ldots,m\}} \; x_i \geq y_j
  \]
  
  \end{definition}

  \begin{definition}[comparator]
  Let $\bar{x} \in \nat^n$ and let $i,j \in \nat$, where $1\leq i < j \leq n$. A comparator is a function $c_{i,j}$ defined as:
  
  \[
    c_{i,j}(\bar{x})=\bar{y} \iff y_i= \max\{x_i,x_j\} \wedge y_j= \min\{x_i,x_j\} \wedge \forall_{k \neq i,j} \; x_k=y_k
  \]
  
  \end{definition}

  \begin{definition}[comparator network]
  We say that $f^n: \nat^n \rightarrow \nat^n$ is a comparator network of order $n$, if it can be represented as the composition of
  finite number of comparators, namely, $f^n=c_{i_1,j_1} \circ \cdots \circ c_{i_k,j_k}$. The size of comparator network
  (number of comparators) is denoted by $|f^n|$. Comparator network of size 0 is denoted by $id^n$.

  \end{definition}

  \begin{wrapfigure}{r}{0.25\textwidth}
      \centering\includegraphics[scale=0.7]{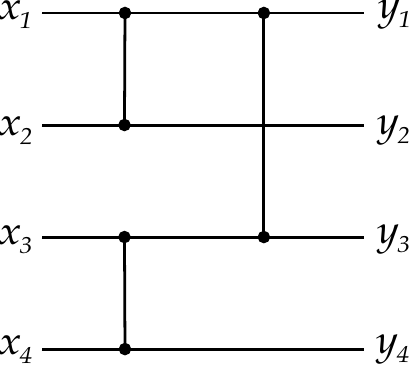}
    \caption{}
    \label{fig:maximum}
  \end{wrapfigure}

  Traditionally comparator networks are presented as circuits that receives $n$ inputs and permutate them using comparators
  connected by "wires". Each comparator has two inputs and two outputs. The "upper" output is the maximum of inputs, and
  "lower" one is minimum. As an example look at Figure \ref{fig:maximum}, where we present a comparator network of order 4, $max^4$, that
  outputs maximum from 4 inputs on its first output, namely, $y_1=\max\{x_1,x_2,x_3,x_4\}$. It is well known
  that $|max^n| = n-1$. We will often omit explicit declaration of order of comparator network when it is not ambiguous.

  \begin{definition}[bitonic sequence]
    A sequence $\bar{x} \in \nat^n$ is a bitonic sequence if $x_1 \leq \ldots \leq x_i \geq \ldots \geq x_n$ 
  for some $i$, where $1 \leq i \leq n$, or a circular shift of such sequence. We distinguish a special case of a bitonic sequence:

  \begin{itemize}
    \item {\em v-shaped}, if $x_1 \geq \ldots \geq x_i \leq \ldots \leq x_n$
  \end{itemize}
  
  \noindent and among v-shaped sequences there are two special cases:
  
  \begin{itemize}
    \item {\em nondecreasing}, if $x_1 \leq \ldots \leq x_n$,
    \item {\em nonincreasing}, if $x_1 \geq \ldots \geq x_n$.
  \end{itemize}
  \label{def:bit}
  \end{definition}

  \begin{definition}[sorting network]
    A comparator network $f^n$ is a {\em sorting network}, if for each $\bar{x} \in \nat^n$, $f^n(\bar{x})$ is sorted.
  \end{definition}

  Two types of sorting networks are of interest to us: {\em odd-even} and {\em pairwise}.
  Based on their ideas, Knuth \cite{knuth} (for odd-even network) and Codish and Zazon-Ivry \cite{pairwise} (for pairwise network)
  showed how to transform them into selection networks (we name them $oe\_sel^n_k$ and $pw\_sel^n_k$ respectively).
  
  \begin{definition}[top $k$ sorted sequence]
    A sequence $\bar{x} \in \nat^n$ is top $k$ sorted, with $k \leq n$, if $\tuple{x_1,\ldots,x_k}$ is sorted and
    $\tuple{x_1,\ldots,x_k} \succeq \tuple{x_{k+1},\ldots,x_n}$.
  \end{definition}
  
  \begin{definition}[selection network]
    A comparator network $f^n_k$ (where $k \leq n$) is a {\em selection network},
    if for each $\bar{x} \in \nat^n$, $f^n_k(\bar{x})$ is top $k$ sorted.
  \end{definition}

  To simplify the presentation we assume that $n$ and $k$ are powers of 2.

  A clause is a disjunction of literals (boolean variables $x$ or their negation $\neg x$). A CNF formula is a conjunction of one or more clauses. 
  
  A unit propagation (UP) is a process, that for given CNF formula, clauses are sought in which all literals but one are false (say $l$) and 
  $l$ is undefined (initially only clauses of size one satisfy this condition). This literal $l$ is set to true and the process is
  iterated until reaching a fix point.

  Cardinality constraints are of the form $x_1 + \ldots + x_n \sim k$, where $k \in \nat$ and $\sim$ belongs to $\{<,\leq,=,\geq,>\}$.
  We will focus on cardinality constraints with less-than relation, i.e. $x_1 + \ldots + x_n < k$.
  An encoding (a CNF formula) of such constraint preserves arc-consistency,
  if as soon as $k-1$ variables among the $x_i$'s become true, the unit propagation sets all other $x_i$'s to false.

  In \cite{card3} authors are using sorting networks for an encoding of cardinality constraints, where inputs and outputs
  of a comparator are boolean variables and comparators are encoded as a CNF formula. In addition, the $k$-th greatest output
  variable $y_k$ of the network is forced to be 0 by adding $\neg y_k$ as a clause to the formula that encodes $x_1 + \ldots + x_n < k$.
  They showed that the encoding preserves arc-consistency.

  A single comparator can be translated to a CNF formula in the following way: let $a$ and $b$ be variables denoting upper and lower
  inputs of the comparator, and $c$ and $d$ be variables denoting upper and lower outputs of a comparator, then:

  \[
  fcomp(a,b,c,d) \Leftrightarrow (c \Leftrightarrow  a \vee b) \wedge ( d \Leftrightarrow  a \wedge b)
  \]

  \noindent is the {\em full encoding} of a comparator. Notice that it consists of 6 clauses. Let $f$ be a comparator network.
  Full encoding $\phi$ of $f$ is a conjunction of full encoding of every comparator of $f$.
  
  In \cite{card4} authors observe that in case of $\sim$ being $<$ or $\leq$,
  it is sufficient to use only 3 clauses for a single comparator, namely:

  \begin{equation}
    hcomp(a,b,c,d) \Leftrightarrow \underbrace{(a \Rightarrow  c)}_{(c1)} \wedge \underbrace{(b \Rightarrow  c)}_{(c2)} \wedge \underbrace{( a \wedge  b  \Rightarrow d)}_{(c3)}
    \label{eq:hcomp}
  \end{equation}

  \noindent We call it: {\em half encoding}.
  In \cite{card4} it is used to translate odd-even sorting network
  to encoding that preserves arc-consistency. We show a more general result (with respect to both \cite{card3} and \cite{card4}), that 
  half encoding of any selection network preserves arc-consistency for the "$<$" and "$\leq$" relations.
  Similar results can be proved for the "$=$" relation using the full encoding of comparators and for the "$>$" or "$\geq$" relations
  using an encoding symmetric to $hcomp(a,b,c,d)$, namely:
  $(d \Rightarrow a) \wedge (d \Rightarrow b) \wedge ( c \Rightarrow a \vee b )$.

  \section{Pairwise and bitonic selection networks}

  Now we present two constructions for selection networks. First, we recall the
  definition of pairwise selection networks by Codish and Zazon-Ivry
  \cite{pairwise}. Secondly, we give the auxiliary construction of a {\em bitonic}
  selection network $bit\_sel^n_k$, that we will use to estimate the sizes of our
  improved pairwise selection network in Section 5.

  \begin{definition}[domination]
    $\bar{x} \in \nat^n$ dominates $\bar{y} \in \nat^n$ if $x_i \geq y_i$ (for $i=1..n$).
  \end{definition}
  
  \begin{definition}[splitter]
    A comparator network $f^n$ is a {\em splitter} if for any sequence $\bar{x}~\in~\nat^n$, if $\bar{y} = f^n(\bar{x})$,
    then $left(\bar{y})$ dominates $right(\bar{y})$.
  \end{definition}

  \begin{observation}
    We can construct splitter $split^n$ by joining inputs $\tuple{i,n/2+i}$, for $i=1..n/2$, with a comparator. Size of a splitter is
    $|split^n| = n/2$.
  \end{observation}
  
  \begin{lemma}
    If $\bar{b} \in \nat^n$ is bitonic and $\bar{y}=split^n(\bar{b})$, then $left(\bar{y})$ and $right(\bar{y})$ are bitonic and
    $left(\bar{y}) \succeq right(\bar{y})$.
    \label{lma:split_bit}
  \end{lemma}
  
  \begin{proof}
    See Appendix B of \cite{batcher}.
  \end{proof}
  
  \begin{net}[$pw\_sel^n_k$; see \cite{pairwise}]
  
  Input: any $\bar{x} \in \nat^n$.

  \begin{enumerate}
    \item If $k=1$, return $max^n(\bar{x})$.
    \item If $k=n$, return $oe\_sort^n(\bar{x})$.
    \item Compute $\bar{y} = split(\bar{x})$.
    \item Compute $\bar{l} = pw\_sel^{n/2}_k(\bar{y})$ and $\bar{r}=pw\_sel^{n/2}_{k/2}(\bar{y})$.
    \item Compute $pw\_merge^n_k(\bar{l} :: \bar{r})$.
  \end{enumerate}
  \label{net:pw}
  \end{net}

  Notice that since we introduced a splitter as the third step, in the recursive calls
  we need to select $k$ top elements from the first half of $\bar{y}$, but only $k/2$
  elements from the second half. The reason: $r_{k/2+1}$ cannot be one of the first
  $k$ largest elements of $\bar{l} :: \bar{r}$. First, $r_{k/2+1}$ is smaller than any one of $\tuple{r_1,\ldots,r_{k/2}}$
  (by the definition of top $k$ sorted sequence), and second, $\tuple{l_1,\ldots,l_{k/2}}$ dominates $\tuple{r_1,\ldots,r_{k/2}}$,
  so $r_{k/2+1}$ is smaller than any one of $\tuple{l_1,\ldots,l_{k/2}}$. From this argument we make the following observation:

  \begin{observation}
    If $\bar{l} \in \nat^{n/2}$ is top $k$ sorted, $\bar{r} \in \nat^{n/2}$ is top $k/2$ sorted and $\tuple{l_1,\ldots,l_{k/2}}$
    dominates $\tuple{r_1,\ldots,r_{k/2}}$, then $k$ largest elements of $\bar{l} :: \bar{r}$ are in
    $\tuple{l_1,\ldots,l_k} :: \tuple{r_1,\ldots,r_{k/2}}$.
    \label{obs:second_step}
  \end{observation}

  The last step of Network \ref{net:pw} merges $k$ top elements from $\bar{l}$ and $k/2$ top elements from $\bar{r}$
  with so called {\em pairwise merger}. We will omit the construction of this merger, because it is not relevant to our work.
  We would only like to note, that its size is: $|pw\_merge^n_k|=k\log k - k + 1$.
  Construction of the merger as well as the detailed proof of correctness
  of network $pw\_sel^n_k$ can be found in Section 6 of \cite{pairwise}.

  \begin{definition}[bitonic splitter]
    A comparator network $f^n$ is a bitonic splitter if for any two sorted sequences $\bar{x},\bar{y}~\in~\nat^{n/2}$,
    if $\bar{z} = bit\_split^n(\bar{x}::\bar{y})$, then (1) $left(\bar{z}) \succeq right(\bar{z})$
    and (2) $left(\bar{z})$ and $right(\bar{z})$ are bitonic.
    \label{def:bitonic_splitter}
  \end{definition}

  \begin{observation}
    We can construct bitonic splitter $bit\_split^n$ by joining inputs $\tuple{i,n-i+1}$, for $i=1..n/2$, with a comparator.
    Size of a bitonic splitter is $|bit\_split^n| = n/2$.
  \end{observation}

  We now present the procedure for construction of the bitonic selection network. We use the odd-even sorting network $oe\_sort$
  and the network $bit\_merge$ (also by Batcher \cite{batcher}) for sorting bitonic sequences as black-boxes. As a reminder: $bit\_merge^n$
  consists of two steps, first we use $\bar{y} = split^n(\bar{x})$,
  then recursively compute $bit\_merge^{n/2}$ for $left(\bar{y})$ and $right(\bar{y})$
  (base case, $n=2$, consists of a single comparator). Size of this network is: $|bit\_merge^n|=n\log n/2$.

  Bitonic selection network $bit\_sel^n_k$ is constructed by the following procedure.

  \begin{net}[$bit\_sel^n_k$]
  Input: any $\bar{x} \in \nat^n$.

  \begin{enumerate}
    \item Let $l=n/k$. Partition input $\bar{x}$ into $l$ consecutive blocks, each of size $k$,
      then sort each block with $oe\_sort^k$, obtaining $B_1,\ldots,B_l$.
    \item While $l>1$, do the following:
      \begin{enumerate}
        \item Collect blocks into pairs $\tuple{B_1,B_2},\ldots,\tuple{B_{l-1},B_{l}}$.
        \item Compute $\bar{y_i} = bit\_split^{2k}(B_i :: B_{i+1})$ for each $i \in \{1,3,\ldots,l-1\}$.
        \item Compute $B_{\lceil i/2 \rceil}' = bit\_merge^k(left(\bar{y_i}))$ for each result of previous step.
        \item Let $l=l/2$. Relabel $B_i'$ to $B_i$, for $1 \leq i \leq l$.
      \end{enumerate}
  \end{enumerate}
  \label{net:bit}
  \end{net}

  \begin{theorem}
    A comparator network $bit\_sel^n_k$ constructed by the procedure Network \ref{net:bit} is a selection network.
  \end{theorem}

  \begin{proof}
    Let $\bar{x} \in \nat^n$ be the input to $bit\_sel^n_k$. After step one we get sorted sequences $B_1,\ldots,B_l$, where $l=n/k$.
    Let $l_m$ be the value of $l$ after $m$ iterations. Let $B^m_1,\ldots,B^m_{l_m}$ be the blocks after $m$ iterations. We will prove
    by induction that:

    \begin{center}
    {\em $P(m)$: if $B_1,\ldots,B_{l}$ are sorted and are containing $k$ largest elements of $\bar{x}$,
      then after $m$-th iteration of the second step: $l_m=l/2^m$, $B^m_1,\ldots,B^m_{l_m}$ are sorted and are containing
      $k$ largest elements of $\bar{x}$.}
    \end{center}

    \noindent If $m=0$, then $l=1$, so $P(m)$ holds. We show that $\forall_{m\geq 0}$ $(P(m) \Rightarrow P(m+1))$.
    Consider $(m+1)$-th iteration of step two. By the induction hypothesis $l_m=l/2^m$, $B^m_1,\ldots,B^m_{l_m}$
    are sorted and are containing $k$ largest elements of $\bar{x}$.  We will show that
    $(m+1)$-th iteration does not remove any element from $k$ largest elements of $\bar{x}$. To see this, notice that if
    $\bar{y_i} = bit\_split^{2k}(B^m_i :: B^m_{i+1})$ (for $i \in \{1,3,\ldots,l_m-1\}$), then $left(\bar{y_i}) \succeq right(\bar{y_i})$ and
    that $left(\bar{y_i})$ is bitonic (by Definition \ref{def:bitonic_splitter}). Because of those two facts, $right(\bar{y_i})$ is
    discarded and $left(\bar{y_i})$ is sorted using $bit\_merge^k$. After this, $l_{m+1}=l_m/2=l/2^{m+1}$ and 
    blocks $B^{m+1}_1,\ldots,B^{m+1}_{l_{m+1}}$ are sorted. Thus $P(m+1)$ is true.

    Since $l=n/k$, then by $P(m)$ we see that the second step will terminate after $m=\log \frac{n}{k}$
    iterations and that $B_1$ is sorted and contains $k$ largest elements of $\bar{x}$.
  \end{proof}

  \begin{figure}[ht]
    \begin{center}
      \includegraphics[scale=0.5]{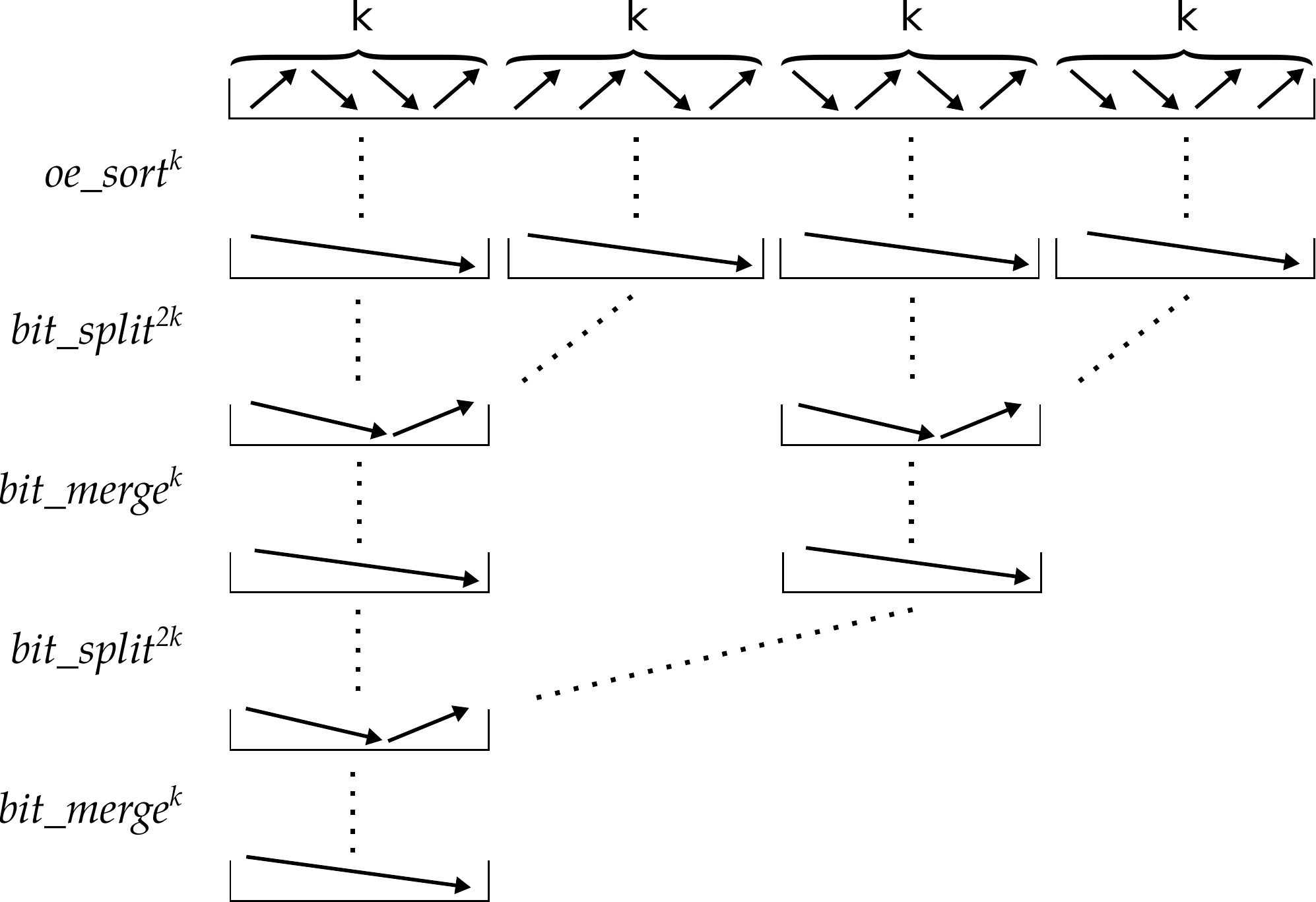}
    \end{center}
    \caption{Bitonic selection network -- schema of construction}
    \label{fig:bit-sel}
  \end{figure}

  Schema of construction of bitonic selection network is shown in Figure \ref{fig:bit-sel}. The size of bitonic selection network is:

  \begin{align}
    |bit\_sel^n_k| &= \frac{n}{k}|oe\_sort^k| + \left(\frac{n}{k}-1\right)(|bit\_split^{2k}|+|bit\_merge^k|) \notag \\ 
    &= \frac{1}{4}n\log^2k + \frac{1}{4}n\log k + 2n - \frac{1}{2}k\log k - k - \frac{n}{k} \label{eq:bit}
  \end{align}

  In Figure \ref{fig:bit-vs-pw} we present bitonic and pairwise selection networks for $n=8$ and $k=2$.

  \begin{figure}[ht]
    \begin{center}
      \subfloat[\label{fig:oe-sel}]{%
      \includegraphics[]{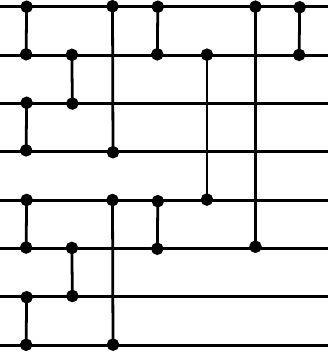}
      }
      \hfill
      \subfloat[\label{fig:pw-sel}]{%
      \includegraphics[]{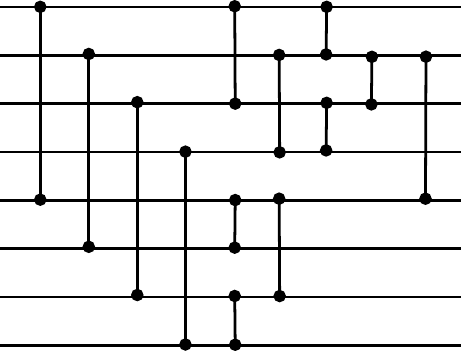}
      }
    \end{center}
    \caption{a) bitonic selection network; b) pairwise selection network; $n=8$, $k=2$.}\label{fig:bit-vs-pw}
  \end{figure}

\section{New Smaller Selection Networks}

  As mentioned in the previous section, only
  the first $k/2$ elements from the second half of the input are relevant when we get to the merging step in $pw\_sel^n_k$.
  We will exploit this fact to create a new, smaller merger.
  We will use the concept of bitonic sequences, therefore the new merger will be called $pw\_bit\_merge^n_k$ 
  and the new selection network: $pw\_bit\_sel^n_k$. The network $pw\_bit\_sel^n_k$ is generated by substituting the last step
  of $pw\_sel^n_k$ with $pw\_bit\_merge^n_k$. The new merger consists of two steps:

  \begin{net}[$pw\_bit\_merge^n_k$]

    Input: $\bar{l} :: \bar{r}$, where $\bar{l} \in \nat^{n/2}$ is top $k$ sorted and $\bar{r} \in \nat^{n/2}$
    is top $k/2$ sorted and $\tuple{l_1,\ldots,l_{k/2}}$ dominates $\tuple{r_1,\ldots,r_{k/2}}$.

    \begin{enumerate}
      \item Compute $\bar{y}=bit\_split^k(l_{k/2+1},\ldots,l_k,r_1,\ldots,r_{k/2})$,
        let $\bar{b}=\tuple{l_1,\ldots,l_{k/2}} :: \tuple{y_1,\ldots,y_{k/2}}$.
      \item Compute $bit\_merge^k(\bar{b})$.
    \end{enumerate}
    \label{net:pw_merge}
  \end{net}
  
  \begin{theorem}
    The output of Network \ref{net:pw_merge} consists of sorted $k$ largest elements from input $\bar{l} :: \bar{r}$, assuming
    that $\bar{l} \in \nat^{n/2}$ is top $k$ sorted and $\bar{r} \in \nat^{n/2}$
    is top $k/2$ sorted and $\tuple{l_1,\ldots,l_{k/2}}$ dominates $\tuple{r_1,\ldots,r_{k/2}}$.
    \label{thm:pw_merge}
  \end{theorem}

  \begin{proof}
    We have to prove two things: (1) $\bar{b}$ is bitonic and (2) $\bar{b}$ consists of $k$ largest elements from $\bar{l} :: \bar{r}$.

    (1) Let $j$ be the last index in the sequence $\tuple{k/2+1,\ldots,k}$, for which $l_j > r_{k-j+1}$. If such $j$ does not exist,
    then $\tuple{y_1,\ldots,y_{k/2}}$ is nondecreasing, hence $\bar{b}$ is bitonic (nondecreasing). Assume that $j$ exists, then
    $\tuple{y_{j-k/2+1},\ldots,y_{k/2}}$ is nondecreasing and $\tuple{y_1,\ldots,y_{k-j}}$ is nonincreasing. Adding the fact that
    $l_{k/2} \geq l_{k/2+1} = y_1$ proves, that $\bar{b}$ is bitonic (v-shaped).

    (2) By Observation \ref{obs:second_step}, it is sufficient to prove that $\bar{b} \succeq \tuple{y_{k/2+1},\ldots,y_k}$.
    Since $\forall_{k/2<j\leq k}$ $l_{k/2} \geq l_j \geq \min\{l_j, r_{k-j+1}\}=y_{3k/2-j+1}$, then
    $\tuple{l_1,\ldots,l_{k/2}} \succeq \tuple{y_{k/2+1},\ldots,y_k}$ and by Definition \ref{def:bitonic_splitter}:
    $\tuple{y_1,\ldots,y_{k/2}} \succeq \tuple{y_{k/2+1},\ldots,y_k}$. 
    Therefore $\bar{b}$ consists of $k$ largest elements from $\bar{l} :: \bar{r}$.

    The bitonic merger in step 2 receives a bitonic sequence, so it outputs a sorted sequence, which completes the proof.
  \end{proof}

  The first step of improved pairwise merger is illustrated in Figure \ref{fig:pw_bit}.
  We use $k/2$ comparators in the first step and $k\log k/2$
  comparators in the second step. We get a merger of size $k\log k/2 + k/2$, which is better
  than the previous approach. In the following it is shown that we can do even better and eliminate $k/2$ term.

  \begin{figure}[ht]
    \begin{center}
      \includegraphics[width=\textwidth]{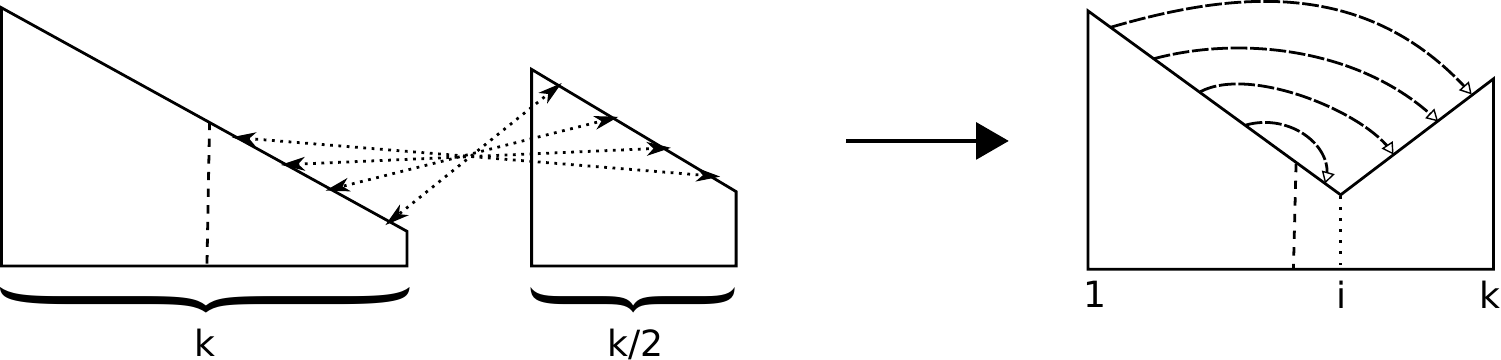}
    \end{center}
    \caption{Making the bitonic sequence. Arrows on the right picture show directions of inequalities.
      Sequence on the right is v-shape s-dominating at point $i$.}
    \label{fig:pw_bit}
  \end{figure}

  The main observation is that the result of the first step of $pw\_bit\_merge$ operation: $\tuple{b_1,b_2,\ldots,b_k}$ is not only bitonic,
  but what we call {\em v-shape s-dominating}.

  \begin{definition}[s-domination]
    A sequence $\bar{b} =\tuple{b_1,b_2,\ldots,b_k}$ is s-dominating if 
    $\forall_{1\leq j \leq k/2}$ $b_j \geq b_{k-j+1}$.
    \label{def:sd}
  \end{definition} 

  \begin{lemma}
    If $\bar{b} =\tuple{b_1,b_2,\ldots,b_k}$ is v-shaped and s-dominating, then $\bar{b}$ is nonincreasing
    or $\exists_{k/2<i<k} \; b_i < b_{i+1}$.
    \label{lma:sd}
  \end{lemma}

  \begin{proof}
    Assume that $\bar{b}$ is not nonincreasing. Then $\exists_{1\leq j<k} \; b_j < b_{j+1}$. Assume that $j\leq k/2$. Since
    $\bar{b}$ is v-shaped, $b_{j+1}$ must be in nondecreasing part of $\bar{b}$. If follows that
    $b_j < b_{j+1} \leq \ldots \leq b_{k/2} \leq \ldots \leq b_{k-j+1}$. That means that $b_j<b_{k-j+1}$. On the other hand, $\bar{b}$
    is s-dominating, thus $b_j \geq b_{k-j+1}$ -- a contradiction.
  \end{proof}
  
  We will say that a sequence $\bar{b}$ is {\em v-shape s-dominating at point $i$} if $i$ is the smallest index greater than
  $k/2$ such that $b_i < b_{i+1}$ or $i=k$ for a nonincreasing sequence.

  \begin{lemma}
    Let $\bar{b}=\tuple{b_1,b_2,\ldots,b_k}$ be v-shape s-dominating at point $i$, then
    $\tuple{b_1,\ldots,b_{k/4}} \succeq \tuple{b_{k/2+1},\ldots,b_{3k/4}}$.
    \label{lma:transit}
  \end{lemma}

  \begin{proof}
    If $\bar{b}$ is nonincreasing, then the lemma holds. From Lemma \ref{lma:sd}: $k/2<i<k$. 
    If $i > 3k/4$, then by Definition \ref{def:bit}: $b_1 \geq \ldots \geq b_{3k/4} \geq \ldots \geq b_i$, so lemma holds.
    If $k/2 < i \leq 3k/4$, then by Definition \ref{def:bit}:  $b_1 \geq \ldots \geq b_i$,
    so $\tuple{b_1,\ldots,b_{k/4}} \succeq \tuple{b_{k/2+1},\ldots,b_{i}}$. Since $b_i < b_{i+1} \leq \ldots \leq b_{3k/4}$,
    it suffices to prove that $b_{k/4} \geq b_{3k/4}$.
    By Definition \ref{def:sd} and \ref{def:bit}: $b_{k/4} \geq b_{3k/4+1} \geq b_{3k/4}$.
  \end{proof}

  \begin{definition}[half splitter]
    A {\em half splitter} is a comparator network constructed by comparing inputs 
    $\tuple{k/4+1,3k/4+1},\ldots,\tuple{k/2,k}$ (normal splitter with first $k/4$ comparators removed).
    We will call it $half\_split^k$.
    \label{def:hs}
  \end{definition}

  \begin{lemma}
    If $\bar{b}$ is v-shape s-dominating, then $half\_split^k(\bar{b})=split^k(\bar{b})$.
    \label{lma:hs_s}
  \end{lemma}

  \begin{proof}
    Directly from Lemma \ref{lma:transit}.
  \end{proof}

  \begin{lemma}
    Let $\bar{b}$ be v-shape s-dominating.
    Following statements are true: (1) $left(half\_split^k(\bar{b}))$ is v-shape s-dominating;
    (2) $right(half\_split^k(\bar{b}))$ is bitonic; (3) $left(half\_split^k(\bar{b})) \succeq right(half\_split^k(\bar{b}))$.
  \label{lma:dom}
  \end{lemma}

  \begin{proof}
    (1) Let $\bar{y}=left(half\_split^k(\bar{b}))$. First we show that $\bar{y}$ is v-shaped.
    If $\bar{y}$ is nonincreasing, then it is v-shaped. Otherwise, let $j$ be the first index
    from the range $\{1,\ldots,k/2\}$, where $y_{j-1}<y_j$. Since $y_j=\max\{b_j,b_{j+k/2}\}$ and
    $y_{j-1} \geq b_{j-1} \geq b_j$, thus $b_j < b_{j+k/2}$.
    Since $\bar{b}$ is v-shaped, element $b_{j+k/2}$ must be in nondecreasing part of $\bar{b}$.
    It follows that $b_{j} \geq \ldots \geq b_{k/2}$ and $b_{j+k/2} \leq \ldots \leq b_k$.
    From this we can see that $\forall_{j\leq j' \leq k/2}$ $y_{j'}=\max\{b_{j'},b_{j'+k/2}\} = b_{j'+k/2}$,
    so $y_j \leq \ldots \leq y_{k/2}$. Therefore $\bar{y}$ is v-shaped.

    Next we show that $\bar{y}$ is s-dominating. Consider any $j$, where $1 \leq j \leq k/4$.
    By Definition \ref{def:bit} and \ref{def:sd}: $b_j \geq b_{k/2-j+1}$ and $b_j \geq b_{k-j+1}$,
    therefore $y_j = b_j \geq \max\{b_{k/2-j+1},b_{k-j+1}\} = y_{k/2-j+1}$, thus proving that $\bar{y}$ is s-dominating.
    Concluding: $\bar{y}$ is v-shape s-dominating.
    
    (2) Let $\bar{z}=right(half\_split^k(\bar{b}))$. By Lemma \ref{lma:hs_s}: $\bar{z}=right(split^k(\bar{b}))$. We know that $\bar{b}$
    is a special case of bitonic sequence, therefore using Lemma \ref{lma:split_bit} we get that $\bar{z}$ is bitonic.

    (3) Let $\bar{w}=half\_split^k(\bar{b})$. By Lemma \ref{lma:hs_s}: $\bar{w}=split^k(\bar{b})$. We know that $\bar{b}$
    is a special case of bitonic sequence, therefore using Lemma \ref{lma:split_bit} we get $left(\bar{w}) \succeq right(\bar{w})$.
  \end{proof}

  Using $half\_split$ and Batcher's $bit\_merge$ and successively
  applying Lemma \ref{lma:dom} to the resulting v-shape
  s-dominating half of the output, we have all the tools needed to construct the improved pairwise merger using half splitters:

  \begin{net}[$pw\_hbit\_merge^n_k$]
    Input: $\bar{l} :: \bar{r}$, where $\bar{l} \in \nat^{n/2}$ is top $k$ sorted and $\bar{r} \in \nat^{n/2}$ is top $k/2$ sorted
    and $\tuple{l_1,\ldots,l_{k/2}}$ dominates $\tuple{r_1,\ldots,r_{k/2}}$.

  \begin{enumerate}
    \item  Compute $\bar{y}=bit\_split^k(l_{k/2+1},\ldots,l_k,r_1,\ldots,r_{k/2})$,
        let $\bar{b}=\tuple{l_1,\ldots,l_{k/2}} :: \tuple{y_1,\ldots,y_{k/2}}$.
    \item Compute $half\_bit\_merge^k(\bar{b})$:
    \begin{enumerate}
      \item If $k=2$, return.
      \item Let $\bar{b'}=half\_split(b_1,\ldots,b_k)$.
      \item Recursively compute $\bar{l'} = half\_bit\_merge^{k/2}(left(\bar{b'}))$.
      \item Compute $\bar{r'} = bit\_merge^{k/2}(right(\bar{b'}))$.
      \item Return $\bar{l'}::\bar{r'}$.
    \end{enumerate}
  \end{enumerate}
  \label{net:pw_merge2}
  \end{net}

  The following theorem states that the construction of $pw\_hbit\_merge^n_k$ is correct.
  
  \begin{theorem}
    The output of Network \ref{net:pw_merge2} consists of sorted $k$ largest elements from input $\bar{l} :: \bar{r}$, assuming
    that $\bar{l} \in \nat^{n/2}$ is top $k$ sorted and $\bar{r} \in \nat^{n/2}$
    is top $k/2$ sorted and $\tuple{l_1,\ldots,l_{k/2}}$ dominates $\tuple{r_1,\ldots,r_{k/2}}$. Also $|pw\_hbit\_merge^n_k|=k \log k/2$. 
    \label{thm:pw_bit_merge}
  \end{theorem}

  \begin{proof}
    Since step 1 in Network \ref{net:pw_merge2} is the same as in Network \ref{net:pw_merge}, we can reuse the proof
    of Theorem \ref{thm:pw_merge} to deduce, that $\bar{b}$ is v-shaped and is containing $k$ largest elements from $\bar{l} :: \bar{r}$.
    Also, since $\forall_{1\leq j \leq k/2}$ $l_j \geq l_{k-j+1}$ and $l_j \geq r_j$, then $b_j = l_j \geq \max\{l_{k-j+1},r_j\} = b_{k-j+1}$,
    so $\bar{b}$ is s-dominating.
    
    We prove by the induction on $k$, that if $\bar{b}$ is v-shape s-dominating, then the sequence $half\_bit\_merge^k(\bar{b})$ is sorted.
    For the base case, consider $k=2$ and a v-shape s-dominating sequence $\tuple{b_1,b_2}$. By Definition \ref{def:sd} this sequence
    is already sorted and we are done. For the induction step, consider $\bar{b'} = half\_split^k(\bar{b})$.
    By Lemma \ref{lma:dom} we get that $left(\bar{b'})$ is v-shape s-dominating and $right(\bar{b'})$ is bitonic. Using the induction
    hypothesis we sort $left(\bar{b'})$ and using bitonic merger we sort $right(\bar{b'})$. By Lemma \ref{lma:dom}:
    $left(\bar{b'}) \succeq right(\bar{b'})$, which completes the proof of correctness. 

    As mentioned in Definition \ref{def:hs}: $half\_split^k$ is just $split^k$ with the first $k/4$ comparators removed.
    So $half\_bit\_merge^k$ is
    just $bit\_merge^k$ with some of the comparators removed. Let's count them: in each level
    of recursion step we take half of comparators from $split^k$ and additional one comparator from the base case ($k=2$).
    We sum them together to get:

    \begin{equation*}
      1 + \sum_{i=0}^{\log k - 2}\frac{k}{2^{i+2}} = 1 + \frac{k}{4}\left(\sum_{i=0}^{\log k - 1}\left(\frac{1}{2}\right)^i - \frac{2}{k}\right) = 1 + \frac{k}{4}\left(2 - \frac{2}{k} - \frac{2}{k} \right) = \frac{k}{2}
    \end{equation*}

   \noindent Therefore we have:
 
    \[
      |pw\_hbit\_merge^n_k| = k/2 + k \log k/2 - k/2 = k \log k/2
    \]

  \end{proof}
  
  The only difference between $pw\_sel$ and our $pw\_hbit\_sel$ is the use of improved merger $pw\_hbit\_merge$ rather than $pw\_merge$.
  By Theorem \ref{thm:pw_bit_merge}, we conclude that $|pw\_merge^n_k| \geq |pw\_hbit\_merge^n_k|$,
  so it follows that:

  \begin{remark}
    $|pw\_hbit\_sel^n_k| \leq |pw\_sel^n_k|$
  \end{remark}

  \section{Sizes of new selection networks}

  In this section we estimate the size of $pw\_hbit\_sel^n_k$. To this end we show that the size of $pw\_hbit\_sel^n_k$
  is upper-bounded by the size of $bit\_sel^n_k$ and use this fact in our estimation. We also compute the exact difference between
  sizes of $pw\_sel^n_k$ and $pw\_hbit\_sel^n_k$ and show that it can be as big as $n\log n/2$. Finally we show graphically how
  much smaller is our selection network on practical values of $n$ and $k$.
  
  We have the recursive formula for the number of comparators of $pw\_hbit\_sel^n_k$:

  \begin{equation}
    |pw\_hbit\_sel^n_k| = \left\{ 
    \begin{array}{l l}
      |pw\_hbit\_sel^{n/2}_k| + |pw\_hbit\_sel^{n/2}_{k/2}|+ &  \\ 
      + |split^n| + |pw\_hbit\_merge^k| & \quad \text{if $k<n$}\\
      |oe\_sort^k| & \quad \text{if $k=n$} \\
      |max^n| & \quad \text{if $k=1$} \\
    \end{array} \right.
  \label{eq:pw}
  \end{equation}

  \begin{lemma}
    $|pw\_hbit\_sel^n_k| \leq |bit\_sel^n_k|$.
    \label{lma:xyz}
  \end{lemma}

  \begin{proof}
    Let $aux\_sel^n_k$ be the comparator network that is generated by substituting recursive calls in $pw\_hbit\_sel^n_k$
    by calls to $bit\_sel^n_k$. Size of this network (for $1<k<n$) is:

    \begin{equation}
      |aux\_sel^n_k| = |bit\_sel^{n/2}_k| + |bit\_sel^{n/2}_{k/2}| + |split^n| + |pw\_hbit\_merge^k|
      \label{eq:aux}
    \end{equation}

    \noindent Lemma \ref{lma:xyz} follows from Lemma \ref{lma:main1} and Lemma \ref{lma:main2} below, where we show that:

    \[
      |pw\_hbit\_sel^n_k| \leq |aux\_sel^n_k| \leq |bit\_sel^n_k|
    \]
  \end{proof}

  \begin{lemma}
    For $1 < k < n$ (both powers of 2), $|aux\_sel^n_k| \leq |bit\_sel^n_k|$.
    \label{lma:main1}
  \end{lemma}

  \begin{proof}
    We compute both values from equations \ref{eq:bit} and \ref{eq:aux}:

    \begin{align*}
      |aux\_sel^n_k|&= \frac{1}{4}n\log^2k + \frac{5}{2}n - \frac{1}{4}k\log k - \frac{5}{4}k - \frac{3n}{2k} \\
      |bit\_sel^n_k|&= \frac{1}{4}n\log^2k + \frac{1}{4}n\log k + 2n - \frac{1}{2}k\log k - k - \frac{n}{k}
    \end{align*}

    \noindent We simplify both sides to get the following inequality:

    \[
      n - \frac{1}{2}k - \frac{n}{k} \leq \frac{1}{2}(n-k)\log k
    \]

    \noindent which can be easily proved by induction.
  \end{proof}

  \begin{lemma}
    For $1 \leq k < n$ (both powers of 2), $|pw\_hbit\_sel^n_k| \leq |aux\_sel^n_k|$.
    \label{lma:main2}
  \end{lemma}

  \begin{proof}
    By induction. For the base case, consider $1=k<n$. If follows by definitions that $|pw\_hbit\_sel^n_k|=|aux\_sel^n_k|=n-1$.
    For the induction step assume that for each $(n',k') \prec (n,k)$ (in lexicographical order) the lemma holds, we get:

    \[
    \begin{tabular}{l r}
      \multicolumn{2}{l}{$|pw\_hbit\_sel^n_k|$} \\
      \multicolumn{2}{l}{$=|pw\_hbit\_sel^{n/2}_{k/2}| + |pw\_hbit\_sel^{n/2}_k|+|split^n|+|pw\_hbit\_merge^k|$} \\
      & \small{\bf (by the definition of $pw\_hbit\_sel$)} \\
      \multicolumn{2}{l}{$\leq|aux\_sel^{n/2}_{k/2}| + |aux\_sel^{n/2}_k|+|split^n|+|pw\_hbit\_merge^k|$} \\
      & \small{\bf (by the induction hypothesis)} \\
      \multicolumn{2}{l}{$\leq|bit\_sel^{n/2}_{k/2}| + |bit\_sel^{n/2}_k|+|split^n|+|pw\_hbit\_merge^k|$} \\
      & \small{\bf (by Lemma \ref{lma:main1})} \\
      \multicolumn{2}{l}{$= |aux\_sel^n_k|$} \\
      & \small{\bf (by the definition of $aux\_sel$)} \\
    \end{tabular}
    \]
  \end{proof}

  Let $N=2^n$ and $K=2^k$. We will compute upper bound for $P(n,k)=|pw\_hbit\_sel^N_K|$ using $B(n,k)=|bit\_sel^N_K|$.
 
  \begin{lemma}
    Let:
    
    \[
      P(n,k,m) = \sum_{i=0}^{m-1}\sum_{j=0}^i\binom{i}{j}\left((k-j)2^{k-j-1}+2^{n-i-1}\right) + \sum_{i=0}^{m}\binom{m}{i}P(n-m,k-i).
    \]

    \noindent Then $\forall_{0 \leq m \leq \min(k,n-k)}$ $P(n,k,m)=P(n,k)$.
    \label{lma:unfold}
  \end{lemma}

  \begin{proof}
    The lemma can be easily proved by induction on $m$.
  \end{proof}

  \begin{lemma}
    $P(n,k,m) \leq 2^{n-2}\left( \left(k-\frac{m}{2}\right)^2 + k + \frac{7m}{4} + 8 \right) + 2^k\left(\frac{3}{2}\right)^m\left(\frac{k}{2} - \frac{m}{6}\right) - 2^k(k+1) - 2^{n-k}\left(\frac{3}{2}\right)^m$.
    \label{lma:bigineq}
  \end{lemma}

\begin{proof}
  First inequality below is a consequence of Lemma \ref{lma:unfold} and \ref{lma:xyz}. We also use the following equations:
  $\sum_{k=0}^n\binom{n}{k}x^{k-1}k = n (1 + x)^{n-1}$, 
  $\sum_{k=0}^n\binom{n}{k}k^2 =  n(n+1)2^{n-2}$,
  $\sum_{k=0}^{n-1}x^{k-1}k = \frac{(1-x)(-nx^{n-1})+(1-x^n)}{(1-x)^2}$.
  
    \begin{align*}
      P(n,k,m) &\leq \underbrace{\sum_{i=0}^{m-1}\sum_{j=0}^i\binom{i}{j}\left((k-j)2^{k-j-1}+2^{n-i-1}\right)}_{(\ref{eq:1})} + \underbrace{\sum_{i=0}^{m}\binom{m}{i}B(n-m,k-i)}_{(\ref{eq:2})}\\
         &= 
         \left(2^{k}\left(\frac{3}{2}\right)^{m}\left(k+1-\frac{m}{3}\right)-2^k(k+1)+m2^{n-1}\right)
          \\
         &\quad+ 2^{n-2}\left( k^2 - km + \frac{m(m-1)}{4} + k + 8 \right) \\
         &\quad+ 2^k\left(\frac{3}{2}\right)^m\left(-\frac{k}{2} + \frac{m}{6} - 1\right) 
         - 
         2^{n-k}\left(\frac{3}{2}\right)^m\\
 &= 2^{n-2}\left( \left(k-\frac{m}{2}\right)^2 + k + \frac{7m}{4} + 8 \right) + 2^k\left(\frac{3}{2}\right)^m\left(\frac{k}{2} - \frac{m}{6}\right) \\
  &\quad- 2^k(k+1) - 2^{n-k}\left(\frac{3}{2}\right)^m
    \end{align*}

    \begin{align}
      \sum_{i=0}^{m-1}\sum_{j=0}^i&\binom{i}{j}\left((k-j)2^{k-j-1}+2^{n-i-1}\right) \label{eq:1} \\
  &= \underbrace{\sum_{i=0}^{m-1}\sum_{j=0}^{i}\binom{i}{j}(k-j)2^{k-j-1}}_{(\ref{eq:1.1})}
  + \underbrace{\sum_{i=0}^{m-1}\sum_{j=0}^{i}\binom{i}{j}2^{n-i-1}}_{(\ref{eq:1.2})} \nonumber \\
  &= \left(2^{k}\middle(\frac{3}{2}\middle)^{m}(k+1-\frac{m}{3})-2^k(k+1)\right) + 
  (m2^{n-1}) \nonumber
    \end{align}

    \begin{align}
      \sum_{i=0}^{m-1}\sum_{j=0}^{i}\binom{i}{j}&(k-j)2^{k-j-1}
      = k2^{k-1}\sum_{i=0}^{m-1}\sum_{j=0}^{i}\binom{i}{j}2^{-j}- 2^{k-1}\sum_{i=0}^{m-1}\sum_{j=0}^{i}\binom{i}{j}2^{-j}j \label{eq:1.1} \\
      &= k2^{k-1}\sum_{i=0}^{m-1}\left(\frac{3}{2}\right)^i - 2^{k-1}\frac{1}{2}\sum_{i=0}^{m-1}\left(\frac{3}{2}\right)^{i-1}i \nonumber \\
      &= k2^{k-1}2\left(\left(\frac{3}{2}\right)^m-1\right) - 2^{k-1}\left(2 - \left(\frac{3}{2}\right)^{m-1}(3-m)\right) \nonumber \\
      &= 2^k\left(\frac{3}{2}\right)^m\left(k + 1 - \frac{m}{3}\right) - 2^k(k+1) \nonumber
    \end{align}

    \begin{equation}
      \sum_{i=0}^{m-1}\sum_{j=0}^{i}\binom{i}{j}2^{n-i-1}=\sum_{i=0}^{m-1}\left(2^{n-i-1}\sum_{j=0}^{i}\binom{i}{j}\right)=\sum_{i=0}^{m-1}2^{n-i-1}2^i=m2^{n-1}
      \label{eq:1.2}
    \end{equation}

    \begin{align}
      &\sum_{i=0}^{m}\binom{m}{i}(2^{n-m-2}(k-i)^2 + 2^{n-m-2}(k-i) - 2^{k-i-1}(k-i) \nonumber \\
      &\quad\quad\quad - 2^{n-m-k+i} + 2^{n-m+1}-2^{k-i}) \label{eq:2} \\
      &= \sum_{i=0}^{m}\binom{m}{i}2^{n-m-2}(k-i)^2 +\sum_{i=0}^{m}\binom{m}{i}2^{n-m-2}(k-i) -\sum_{i=0}^{m}\binom{m}{i}2^{k-i-1}(k-i)\nonumber \\
      &\quad- \sum_{i=0}^{m}\binom{m}{i}2^{n-m-k+i}+\sum_{i=0}^{m}\binom{m}{i}2^{n-m+1} - \sum_{i=0}^{m}\binom{m}{i}2^{k-i}\nonumber \\
      &= 2^{n-m-2}(k^22^m-km2^m+m(m+1)2^{m-2})+2^{n-m-2}(k2^m-m2^{m-1}) \nonumber \\
      &\quad- 2^{k-1}\left(k\left(\frac{3}{2}\right)^m-2^{-m}3^{m-1}m\right)- 2^{n-m-k}3^m + 2^{n+1} - 2^k\left(\frac{3}{2}\right)^m \nonumber \\
      &= 2^{n-2}\left( k^2 - km + \frac{m(m-1)}{4} + k + 8 \right) \nonumber \\
      &\quad+ 2^k\left(\frac{3}{2}\right)^m\left(-\frac{k}{2} + \frac{m}{6} - 1\right) - 2^{n-k}\left(\frac{3}{2}\right)^m \nonumber
    \end{align}

\end{proof}

\begin{theorem}
  For $m= \min(k,n-k)$, $P(n,k) \leq 2^{n-2}\left( \left(k - \frac{m}{2} - \frac{7}{4}\right)^2 + \frac{9k}{2} + \frac{79}{16} \right)$ $+ 2^k\left(\frac{3}{2}\right)^m\left(\frac{k}{2} - \frac{m}{6}\right) - 2^k(k+1) - 2^{n-k}\left(\frac{3}{2}\right)^m$.
  \label{thm:upper-bound}
\end{theorem}

\begin{proof}
  Directly from Lemmas \ref{lma:unfold} and \ref{lma:bigineq}.
\end{proof}

  We will now present the {\em size difference} $SD(n,k)$ between pairwise selection network and our network. Merging step in
  $pw\_sel^N_K$ costs $2^{k}k - 2^k + 1$ and in $pw\_hbit\_sel^N_K$: $2^{k-1}k$, so the difference is given by the following equation:

  \begin{equation}
    SD(n,k) = \left\{ 
    \begin{array}{l l}
      0 & \quad \text{if $n=k$} \\
      0 & \quad \text{if $k=0$} \\
      2^{k-1}k - 2^k + 1 + & \\
      + SD(n-1,k) + SD(n-1,k-1) & \quad \text{if $0<k<n$}
    \end{array} \right.
  \label{eq:pw_size_diff}
  \end{equation}

  \begin{theorem}
    Let $S_{n,k}=\sum_{j=0}^k\binom{n-k+j}{j}2^{k-j}$. Then:

    \[
      SD(n,k) = \binom{n}{k}\frac{n+1}{2} - S_{n,k}\frac{n-2k+1}{2} - 2^k(k-1)-1
    \]
  \end{theorem}

  \begin{proof}
    By straightforward calculation one can verify that $S_{n,0} = 1$, $S_{n,n}
    = 2^{n+1} - 1, S_{n-1,k-1} = \frac{1}{2}(S_{n,k} - \binom{n}{k})$ and
    $S_{n-1,k-1} + S_{n-1,k} = S_{n,k}$. It follows that the theorem is true
    for $k = 0$ and $k = n$. We prove the theorem by induction on pairs
    $(k,n)$. Take any $(k,n)$, $0 < k < n$, and assume that theorem holds for
    every $(k',n') \prec (k,n)$ (in lexicographical order). Then we have:

    \begin{align*}
      SD(n,k)&= 2^{k-1}k - 2^k + 1 + SD(n-1,k) + SD(n-1,k-1) \\
      &= 2^{k-1}k - 2^k + 1 + \binom{n-1}{k}\frac{n}{2} + \binom{n-1}{k-1}\frac{n}{2} - 2^k(k-1)-1 \\
      &\quad -2^{k-1}(k-2)-1 - (S_{n-1,k}\frac{n-2k}{2} + S_{n-1,k-1}\frac{n-2k+2}{2}) \\
      &= \binom{n}{k}\frac{n}{2} - S_{n,k}\frac{n-2k}{2} - S_{n-1,k-1} - 2^k(k-1)-1\\
      &= \binom{n}{k}\frac{n+1}{2}- S_{n,k}\frac{n-2k+1}{2} - 2^k(k-1)-1
    \end{align*}
  \end{proof}
  
\begin{corollary}
  $|pw\_sel^N_{N/2}| - |pw\_hbit\_sel^N_{N/2}| = N\frac{\log N - 4}{2} + \log N + 2$, for $N=2^n$.
\end{corollary}


  Plots in figure \ref{fig:prd} show how much $pw\_sel$ and the upper bound from Theorem \ref{thm:upper-bound}
  are worse than $pw\_hbit\_sel$.
  Lines labeled {\em codish} are plotted from $(|pw\_sel^N_K|-|pw\_hbit\_sel^N_K|)/|pw\_hbit\_sel^N_K|$ and the ones labeled {\em upper}
  are plotted from the formula $(|upper^N_K|-|pw\_hbit\_sel^N_K|)/|pw\_hbit\_sel^N_K|$, where $|upper^N_K|$ is the upper bound from
  Theorem \ref{thm:upper-bound}. Both $|pw\_sel^N_K|$ and $|pw\_hbit\_sel^N_K|$ were computed directly from recursive formulas.
  We can see that we save the most number of comparators when $k$ is larger than $n/2$, nevertheless for small values of $n$ superiority
  of our network is apparent for any $k$. As for the upper bound, it gives a good approximation of $|pw\_hbit\_sel^N_K|$ when $n$ is small
  , but for larger values of $n$ it becomes less satisfactory.

  \begin{figure}[ht]
    \begin{center}
      \subfloat[$N=2^7$\label{fig:prd1}]{%
        \includegraphics[width=0.31\textwidth]{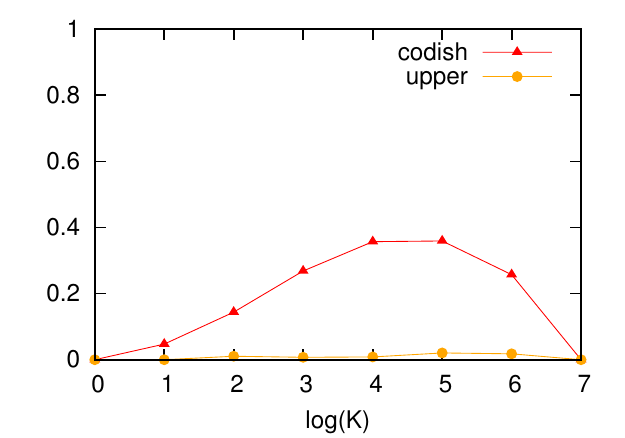}
      }
      ~
      \subfloat[$N=2^{15}$\label{fig:prd2}]{%
      \includegraphics[width=0.31\textwidth]{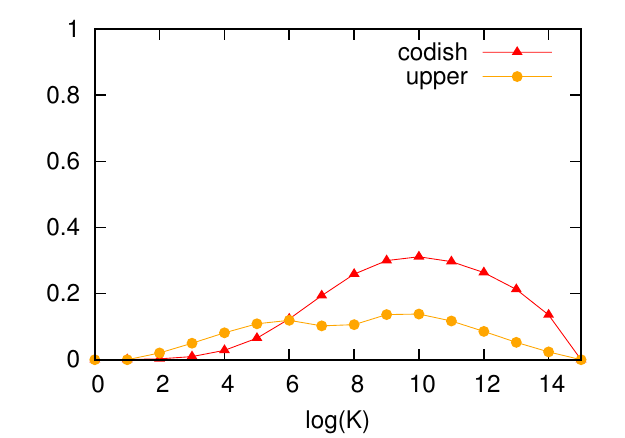}
      }
      ~
      \subfloat[$N=2^{31}$\label{fig:prd3}]{%
      \includegraphics[width=0.31\textwidth]{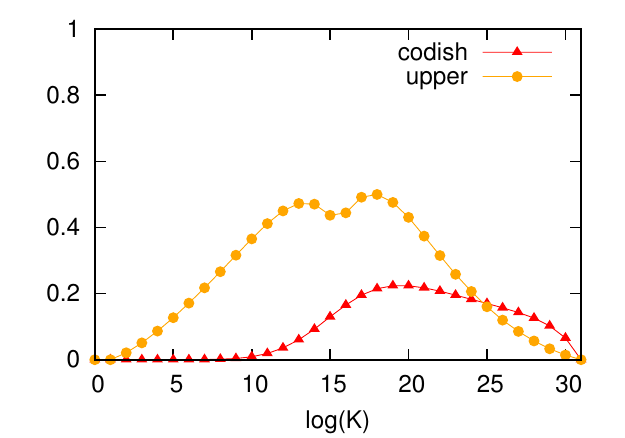}
      }
    \end{center}
    \caption{Comparison of pairwise selection networks for pratical values of $n$ and $k$.}
    \label{fig:prd}
  \end{figure}
  
\section{Arc-consistency of selection networks}

In this section we prove that half encoding of any selection network preserves
arc-consistency with respect to "less-than" cardinality constraints. The proof can
be generalized to other types of cardinality constraints.

  We introduce the convention, that $\tuple{x_1,\ldots,x_n}$ will denote the input and 
  $\tuple{y_1,\ldots,y_n}$ will denote the output of some order $n$ comparator network. We would also like to view them as
  sequences of boolean variables, that can be set to either true ($1$), false ($0$) or undefined ($X$).

  From now on we assume that every network $f$
  is half encoded and when we say "comparator" or "network", we view it in terms of CNF formulas.
  We denote $V[\phi(f)]$ to be the set of variables in encoding $\phi(f)$.

  \begin{observation}
    A single comparator $hcomp(a,b,c,d)$ has the following {\em propagation properties}:

    \begin{enumerate}
      \item If $a=1$ or $b=1$, then UP sets $c=1$ (by $\ref{eq:hcomp}.c1$ or $\ref{eq:hcomp}.c2$).
      \item If $a=b=1$, then UP sets $c=d=1$ (by $\ref{eq:hcomp}.c1$ and $\ref{eq:hcomp}.c3$).
      \item If $c=0$, then UP sets $a=b=0$ (by $\ref{eq:hcomp}.c1$ and $\ref{eq:hcomp}.c2$).
      \item If $b=1$ and $d=0$, then UP sets $a=0$ (by $\ref{eq:hcomp}.c3$).
      \item If $a=1$ and $d=0$, then UP sets $b=0$ (by $\ref{eq:hcomp}.c3$).
    \end{enumerate}
    
    \label{obs:pprop}
  \end{observation}
  
  \begin{lemma}
    Let $f^n_k$ be a selection network. Assume that $k-1$ inputs are set to $1$, and rest of the variables are undefined.
    Unit propagation will set variables $y_1,\ldots,y_{k-1}$ to $1$.
    \label{lma:fprop}
  \end{lemma}

  \begin{proof}
    From propagation properties of $hcomp(a,b,c,d)$ we can see that if comparator receives two $1$s, then it outputs two $1$s,
    when it receives $1$ on one input and $X$ on the other, then it outputs $1$ on the upper output and $X$ on the lower output. From this
    we conclude that a single comparator will sort its inputs, as long as one of the inputs is set to $1$.
    No $1$ is lost, so they must all reach the outputs. Because the comparators comprise a selection network, the $1$s will appear
    at outputs $y_1,\ldots,y_{k-1}$.
  \end{proof}

  The process of propagating $1$s we call a {\em forward propagation}. For the remainder of this section assume that:
  $f^n_k$ is a selection network; $k-1$ inputs are set to $1$, and the rest of the variables are undefined; forward propagation
  has been performed resulting in $y_1,\ldots,y_{k-1}$ to be set to $1$.

  \begin{definition}[path]
    A {\em path} is a sequence of boolean variables
    $\tuple{z_1,\ldots,z_m}$ such that $\forall_{1\leq i \leq m} z_i \in V[\phi(f^n_k)]$ and for all $1 \leq i < m$ there exists
    a comparator $hcomp(a,b,c,d)$ in $\phi(f^n_k)$ for which $z_i \in \{a,b\}$ and $z_{i+1} \in \{c,d\}$.
  \end{definition}
  
  \begin{definition}[propagation path]
    Let $x$ be an undefined input variable. A path $\bar{z}_x =\tuple{z_1,\ldots,z_m}$ $(m \geq 1)$
    is a {\em propagation path}, if $z_1\equiv x$
    and $\tuple{z_2,\ldots,z_m}$ is the sequence of variables that would be set to $1$ by UP, if we would set $z_1=1$.
  \end{definition}

  \begin{lemma}
    If $\bar{z}_x=\tuple{z_1,\ldots,z_m}$ is a propagation path for an undefined variable $x$, then $z_m \equiv y_k$.
    \label{lma:prpp}
  \end{lemma}

  \begin{proof}
    Remember that all $y_1,\ldots,y_{k-1}$ are set to $1$. Setting any undefined input variable $x$ to
    $1$ will result in UP to set $y_k$ to $1$.
    Otherwise $f^n_k$ would not be a selection network.
  \end{proof}

  The following lemma shows that propagation paths are deterministic.
  
  \begin{lemma}
    Let $\bar{z}_x=\tuple{z_1,\ldots,z_m}$ be a propagation path. For each $1 \leq i \leq m$ and $z'_1 \equiv z_i$, if
    $\tuple{z'_1,\ldots,z'_{m'}}$ is a path that would be set to $1$ by UP if we would set $z'_1=1$, then
    $\tuple{z'_1,\ldots,z'_{m'}} = \tuple{z_i,\ldots,z_m}$.
    \label{lma:hier}
  \end{lemma}

  \begin{proof}
    By induction on $l=m-i$. If $l=0$, then $z'_1 \equiv z_m \equiv y_k$ (by
    Lemma \ref{lma:prpp}), so the lemma holds. Let $l \geq 0$ and assume that the
    lemma is true for $z_l$. Consider $z'_1 \equiv z_{l-1} \equiv z_{m-i-1}$. Set
    $z_{m-i-1}=1$ and use UP to set $z_{m-i}=1$. Notice that $z_{m-i} \equiv
    z'_2$, otherwise there would exist a comparator $hcomp(a,b,c,d)$, for which
    $z_{m-i-1}$ is equivalent to either $a$ or $b$ and $z_{m-i} \equiv c$ and
    $z'_2 \equiv d$ (or vice versa). That would mean that a single $1$ on the
    input produces two $1$s on the outputs. This contradicts our reasoning in the
    proof of Lemma \ref{lma:fprop}. By the induction hypothesis
    $\tuple{z'_2,\ldots,z'_{m'}} = \tuple{z_{m-i},\ldots,z_m}$, so
    $\tuple{z'_1,\ldots,z'_{m'}} = \tuple{z_{m-i-1},\ldots,z_m}$.
  \end{proof}
  
  For each undefined input variable $x$ and propagation path $\bar{z}_x=\tuple{z_1,\ldots,z_m}$ we define a directed graph
  $P_x=\{\tuple{z_i,z_{i+1}} \, : \, 1 \leq i < m\}$.  

  \begin{lemma}
    Let $\{x_{i_1},\ldots,x_{i_t}\}$ ($t>0$) be the set of undefined input variables. Then $T = P_{x_{i_1}} \cup \ldots \cup P_{x_{i_t}}$
    is the tree rooted at $y_k$.
  \end{lemma}

  \begin{proof}
    By induction on $t$. If $t=1$, then $T=P_{x_{i_1}}$ and by Lemma \ref{lma:prpp}, $P_{x_{i_1}}$ ends in $y_k$,
    so the lemma holds. Let $t>0$ and assume that the lemma is true for $t$. We will 
    show that it is true for $t+1$. Consider $T=P_{x_{i_1}} \cup \ldots \cup P_{x_{i_{t}}} \cup P_{x_{i_{t+1}}}$.
    By the induction hypothesis $T'=P_{x_{i_1}} \cup \ldots \cup P_{x_{i_{t}}}$ is a tree rooted at $y_k$. By Lemma \ref{lma:prpp},
    $V(P_{x_{i_{t+1}}}) \cap V(T') \neq \emptyset$. Let $z \in V(P_{x_{i_{t+1}}})$ be the first variable, such that $z \in V(T')$. Since
    $z \in V(T')$, there exists $j$ ($1\leq j \leq  t$) such that $z \in P_{x_{i_j}}$. By Lemma \ref{lma:hier},
    starting from variable $z$, paths $P_{x_{i_{t+1}}}$ and $P_{x_{i_j}}$ are identical.
  \end{proof}

  Graph $T$ from the above lemma will be called a {\em propagation tree}.
  
  \begin{theorem}
    If we set $y_k=0$, then unit
    propagation will set all undefined input variables to $0$.
  \end{theorem}

  \begin{proof}
    Let $T$ be the propagation tree rooted at $y_k$.
    We prove by induction on the height $h$ of $T$, that (*) if we set root of $T$ to $0$, then all nodes of the tree will be set to $0$,
    thus all undefined input variables will also be set to $0$. If $h=0$, then $V=\{y_k\}$, so (*) is trivially true. Let $h>0$ and assume
    that (*) holds. We will show that (*) holds for height $h+1$.
    Let $T'$ be the propagation tree of height $h+1$ and let $r=0$ be the root. Consider children of $r$
    in $T'$ and a comparator $hcomp(a,b,c,d)$ for which $r \in \{c,d\}$:

    Case 1: $r$ has two children. The only case is when $r\equiv c=0$. Unit propagation sets $a=b=0$.
    Nodes $a$ and $b$ are roots of propagation trees of height $h$ and are set to $0$,
    therefore by the induction hypothesis all nodes in $T'$ will be set to $0$.

    Case 2: $r$ has one child. Consider two cases: (i) if $r\equiv c=0$ and either $a$ or $b$ is the child of $r$, then UP
    sets $a=b=0$ and either $a$ or $b$ is the root of propagation tree of height $h$ and is set to $0$,
    therefore by the induction hypothesis all nodes in $T'$ will be set to $0$, (ii) $r \equiv d=0$
    and either $a=c=1$ and $b$ is the child of $r$ or
    $b=c=1$ and $a$ is the child of $r$. Both of them will be set to $0$ by UP and again we get the root of
    propagation tree of height $h$ that is set to $0$, therefore by the induction hypothesis all nodes in $T'$ will be set to $0$.
  \end{proof}

\section{Conclusions}

We have constructed a new family of selection networks, which are based on the
pairwise selection ones, but require less comparators to merge subsequences. The
difference in sizes grows with $k$ and is equal to $n\frac{\log n - 4}{2} + \log n
+ 2$ for $k = n/2$. In addition, we have shown that any selection network encoded
in a standard way to a CNF formula preserves arc-consistency with respect to a
corresponding cardinality constraint. This property is important, as many
SAT-solvers take advantage of arc-consistency, making the computation
significantly faster.

It's also worth noting that using encodings based on selection networks give an
extra edge in solving optimization problems for which we need to solve a sequence
of problems that differ only in the decreasing bound of a cardinality constraint.
In this setting we only need to add one more clause $\neg y_{k}$ for a new value
of $k$, and the search can be resumed keeping all previous clauses as it is. This
works because if a comparator network is a $k$-selection network, then it is also
a $k'$-selection network for any $k'<k$. This property is called {\em incremental
strengthening} and most state-of-the-art SAT-solvers provide a user interface for
doing this.

\end{document}